\documentclass[11pt]{article}
\usepackage{amssymb}
\usepackage{amsthm}
\usepackage{amsmath}
\usepackage{fullpage}
\usepackage{graphicx,url,amsmath,amsthm,amsfonts,amssymb,floatflt,subfigure,bbm}
\usepackage{comment}

\newtheorem{theorem}{Theorem}[section]

\newtheorem{lemma}[theorem]{Lemma}

\newtheorem{claim}[theorem]{Claim}
\newtheorem{cor}[theorem]{Corollary}

\newtheorem{definition}[theorem]{Definition}

\newtheorem{remark}[theorem]{Remark}

\newcommand\remove[1]{}

\newcommand{\La}{\mathcal{L}}

\newcommand{\mprod}{\bullet}
\DeclareMathOperator{\tra}{tr}
\DeclareMathOperator{\im}{Im}

\begin{document}
\title{Subgraph Sparsification and Nearly Optimal Ultrasparsifiers}
\author{Alexandra Kolla\thanks{School of Mathematics, Institute for Advanced Study. Research supported by NSF Grant CCF-0832797.}
\and Yury Makarychev\thanks{Toyota Technological Institute at
Chicago.} \and Amin Saberi \thanks{Department of Management Science
and Engineering, Stanford University.} \and Shang-Hua Teng
\thanks{Computer Science Department, University of
  Southern California. Supported by NSF grant CCF-0635102.}}

\date{}

\maketitle

\begin{abstract}
We consider a variation of the spectral sparsification problem where
  we are required to keep a subgraph of the original graph.
Formally, given a union of two weighted graphs $G$ and $W$ and an
integer $k$, we are asked to find a $k$-edge weighted graph $W_k$
such that $G+W_k$ is a good spectral sparsifer of $G+W$. We will
refer to this  problem as the subgraph (spectral) sparsification. We
present a nontrivial condition on $G$ and $W$ such that a good
sparsifier exists and give a polynomial time algorithm to find the
sparsifer.

As a significant application of our technique, we show that for each
positive integer $k$, every $n$-vertex weighted graph has an
$(n-1+k)$-edge spectral sparsifier with relative condition number at
most $\frac{n}{k} \log n\, \tilde{O}(\log\log n)$ where $\tilde{O}()$
hides lower order terms. Our bound is within a factor of
$\tilde{O}(\log \log n)$ from optimal. This nearly settles a
question left open by Spielman and Teng about ultrasparsifiers,
which is a key component in their nearly linear-time algorithms for
solving diagonally dominant symmetric linear systems.

We also present another application of our technique to spectral
optimization in which the goal is to maximize the algebraic
connectivity of a graph (e.g. turn it into an expander) with a
limited number of edges.

\end{abstract}
\thispagestyle{empty}
\newpage
\setcounter{page}{1}
\section{Introduction}

Sparsification is an important technique for designing efficient
  graph algorithms, especially for dense graphs.
Informally, a graph $\tilde{G}$ is a sparsifer of $G$ if they are
  similar in a particular measure (which is important to the
  application that one has in mind), and  that $\tilde{G}$ has linear or nearly
  linear number of edges.
Various notions of graph approximation and sparsification have been
  considered in the literature.
For example, Chew's \cite{Che86} spanners (for shortest path
planning) have the
  property that the distance between every pair of vertices in $\tilde{G}$ is approximately the same as in G.
Benczur and Karger's \cite{BK96} cut-sparsifiers (for cuts and
flows)
  have the property that the weight of the boundary of
  every set of vertices is approximately the same in $G$ as in $\tilde{G}$.

In this paper, we will mainly be interested in the spectral notion
of
  graph similarity introduced by Spielman and Teng \cite{ST04},
  \cite{ST08b}: we say that a weighted undirected graph $H$ is a $\kappa$-approximation
  of another $G$ if for all $ x\in \mathbf{R}^V$,
\begin{equation}\label{eq:patch:sparsif}
x^T\La_G x \leq x^T\La_{\tilde{G}} x\leq \kappa x^T\La_G  x
\end{equation}
where for a weighted undirected  graph $G$,  $\La_G$ is the {\em
Laplacian matrix} of $G$ defined as the following:  For each
$\La_G(i,i)$ is equal to the sum of  weights  of all edges incident
to vertex $i$ and for $i \neq j$, $\La_G(i,j) = -w_{i,j}$, where
$w_{i,j}$ is the weight on edge $(i,j)$.

In \cite{ST04,ST08b}, the following spectral sparsification problem
is
  considered.
Given a weighted graph $G = (V,E,w)$, an integer
  $\tilde{m}\leq |E|$, and $\kappa \geq 1$, find a graph
  $\tilde{G} =\{V, \tilde{E}, \tilde{w}\}$ such that $|\tilde{E}| \leq \tilde{m}$
  and $\tilde{G}$ is a $\kappa$-approximation of $G$.
We will refer to
  this problem and its corresponding optimization problem as  the  {\sc
  Spectral Sparsification}.
Spielman and Teng showed that every weighted graph has a nearly
  linear-sized spectral sparsifier and gave a nearly linear-time
  algorithm for computing  such a sparsifier.
Recently, Batson, Spielman, and Srivastava \cite{BSS} gave a
  beautiful, polynomial-time construction to produce a linear-sized
  spectral sparsifier.

In this paper, we introduce a variation of the spectral
sparsification
  problem which we will refer to as the {\sc Subgraph Sparsification.}
In our version, we are given two weighted graphs $G$ and $W$, an
  integer $k$ and $\kappa\geq 1$.
The goal is to find a $k$-edge weighted graph $W_k$ such
  that $(G+W_k)$ is a $\kappa$-approximation of $(G+W)$.
The challenge in the new version of the sparsification problem is
  that we have to respect part of the graph, i.e., $G$, and only
  modify part of graph given in $W$.

As the main technical contribution of the paper, we give a
nontrivial condition about $G$ and $W$ such that a good sparsifier
exists. Our proof critically uses  the intuition of Batson,
Spielman, and  Srivastava~\cite{BSS}, that uses  potential
functions that guide an incremental process for selecting the edges
of the sparisifier. We will refer to that as as the {\em BSS
process}. We have enhanced their approach with new
  understanding about subspace sparsification and spectral
  approximation.

Our challenge, at high level, is the following. The BSS process uses
two carefully chosen barriers (see Section 2) so that at each step,
all eigenvalues can be kept far enough from these
  barriers. They have $\Theta(n)$ edges to select. So they consider the entire
$n$-dimensional space and have step size $\Theta(1/n)$ on these
barriers.

On the other hand, we can only add $k$ edges, where $k$ can be
arbitrarily smaller than $n$. The addition of each edge can only
increase smallest eigenvalue to the second smallest eigenvalue.
Therefore the addition of $k$ edges can only improve the subspace
defined by the $k$ smallest eigenvalue. Now, the critical part of
the argument is that to build a good sparsifier,  we need to ensure
that the addition of the edges does not increase the high spectra by
too much. So in our incremental process, we need to keep track of
two subspaces, a fixed one defined by the $k$ smallest eigenvalues
and a floating one defined by the higher spectra.

We developed an analysis for performing spectral analysis in the
projection of a sequence of two subspaces, which might be
interesting   on its own right. Our analysis also provide a nice
example for using majorization.

Our ability to conduct sparsification on a subgraph enables us to
obtain improved results for a few problems on spectral optimization.
The first application that we consider is the problem of finding
ultrasparsifiers as defined in Spielman and Teng \cite{ST04}. For
parameters $\kappa \geq 1$ and   $k\geq 1$, a weighted undirected
graph $U$ is a {\em $(\kappa,k)$-ultrasparsifier} of another graph
$G$, if $U$ has at most $n-1+k$  edges, and $\La_U \preceq \La_G \preceq
\kappa \cdot \La_U$. Ultrasarsifiers are essential in the application of
the preconditioning   techniques for solving linear systems
\cite{ST04,ST08b}. It has been shown in \cite{ST04} that every
weighted undirected graph   $G$ has a $(\frac{n}{k} \log ^{O(1)}
n,k)$ ultrasparsifiers, for  any $k$. As a significant application of
our subgraph sparsification technique, we show that for every
  positive integer $k$, every $n$-vertex weighted graph has a $(\frac{n}{k}\log
n\,\tilde{O}(\log \log n),k)$-ultrasparsifier. Our bound is within a
factor of  $\tilde{O}(\log \log)$ from    the optimal. This new
result nearly settles a question about ultrasparsifiers left open
by Spielman and Teng.

At high level, our solution to ultrasparsification is quite
  simple, once we have our subgraph sparsification result.
Given a weighted graph$G$, we first construct a low-stretch spanning
  tree \cite{AKPW,EEST,ABN08} $T$ of $G$.
We then apply an elegant result of Spielman and Woo \cite{SW} which
states that the sum of the relative condition numbers of $\La_G$ and
$\La_T$ is equal to the total stretch to embed $G$ onto $T$. We will
also use Spielman--Woo's tail distribution bound on the number of
relative
  eigenvalues of $\La_G$ and $\La_T$ that are larger than a given
  parameter.

Algorithmically, we start with the best available \cite{ABN08}
low-stretch spanning tree $T$ of $G$
  whose total stretch is $n\log n\,\tilde{O}(\log \log n)$.
We then
  consider the subgraph sparsification problem defined by $T$
  and $W = \frac{k}{n\log n \tilde{O}(\log \log n )} G$.
We apply the structure theorem of Spielman and Woo \cite{SW} to show
that
  $(T,W)$ satisfy our condition for subgraph sparsification and
  apply our result to show that there exists  a $k$-edge weighted graph
  $W_k$ whose edges are in $W$ such that $T+W_k$ is a spectral
  approximation of $T+W$.
It is then not hard to prove that $T+W_k$ is an
  a $(\frac{n}{k}\log n\,\tilde{O}(\log \log n),k)$-ultrasparsifier.

As another application of our technique on subgraph sparsification,
we consider the following spectral optimization problem studied in
\cite{BG}: Given a graph $G$ and a parameter $k$, we are asked to
find $k$ edges amongst a set of candidate edges
  to add to $G$ so as to maximize its algebraic connectivity.
Algebraic connectivity has emerged as an important parameter for
measuring the robustness and stability of a network and is an
essential factor in the performance of various search, routing and
information diffusion algorithms.

The spectral optimization considered in this paper is known to be
NP-hard~\cite{damon} and no approximation guarantee for it was known
prior to our work. We give an SDP-based approximation algorithm for
the problem. Our techniques for subgraph sparsification enable us to
develop a novel \textit{rounding} scheme in order to find a
combinatorial solution. Since the integrality gap of the SDP is
unbounded, our analysis involves adding a separate upper bound,
which is roughly the $k$-th largest eigenvalue of the Laplacian of
$G$ to approximate the optimum solution.
%

\section{Preliminaries}
\textbf{Matrix Notation and Definitions.} We denote the Laplacian of a graph $G$ by $\La_G$.
For brevity, we write ${G_1} \preceq G_2$ to denote~$\La_{G_1} \preceq \La_{G_2}$.
For an $n\times n$ matrix $A$, let $\lambda_{\mathrm{min}}(A) \equiv \lambda_1(A) \leq \lambda_2(A) \leq \dots \leq \lambda_n(A) \equiv \lambda_{\mathrm{max}}(A)$
be the set of eigenvalues in the increasing order.
Let $A^\dagger$ be the
pseudoinverse of $A$. If $A$ is symmetric, $A^\dagger$ is also symmetric and
$AA^\dagger = A^\dagger A = P_{\im(A)}$, where  $P_{\im(A)}$ is the orthogonal projection on $\im(A)$.
Let $A\mprod B \equiv \tra A^TB$ be the Frobenius product
of matrices $A$ and $B$. We define the \textit{condition number} of a non-singular matrix $A$ as
$\kappa=\|A\|\|A^{-1}\|$,
which is equal to $\lambda_{\mathrm{max}}(A)/\lambda_{\mathrm{min}}(A)$ if $A$ is a (symmetric) positive
definite matrix. For positive definite
matrices $A$, $B$ with $\im A = \im B$, we define the relative condition number as
$$\kappa(A,B) =\max_{x\notin\ker B} \frac{x^T A x}{x^T B x}  \cdot \max_{x\notin\ker A} \frac{x^TBx}{x^TAx}.$$

\paragraph{Ultrasparsifiers.} We say that a graph is $k$--ultra-sparse if it has at most
$n-1+k$ edges. We note that a spanning tree is $0$--ultra-sparse. A ($\kappa,k)$ \textit{ultra-sparsifier} of a graph
$G = (V,E,w)$ is a $k$--ultra-sparse subgraph of $G$ such that
 $U \preceq G \preceq \kappa \cdot U$~\cite{ST04}.

\section{Matrix Sparsifiers} \label{sec:patch:patch}
In this section, we prove an analog of the
sparsification theorem of Batson, Spielman, and Srivastava~\cite{BSS}.
\begin{definition}\label{def:patch}(\textbf{Graph Patch})
Let $G$ be a (weighted) graph. A graph $W$ on the vertices of $G$ is a $(k,T,\lambda^*)$-patch
for $G$ if the following properties hold\footnote{we have $\lambda_{k+1}
(\La_G \La_{G+W}^{\dagger}) =
\lambda_{k+1}((\La_{G+W}^{\dagger})^{1/2}\La_G(\La_{G+W}^{\dagger})^{1/2})$,
since $\lambda_i(AB) = \lambda_i(BA)$ for every two square matrices $A$ and $B$},
\begin{enumerate}
\item $\lambda_{k+1}(\La_G \La_{G+W}^{\dagger}) \equiv
\lambda_{k+1}((\La_{G+W}^{\dagger})^{1/2}\La_G(\La_{G+W}^{\dagger})^{1/2})
\geq \lambda^*$;
\item $\tra(\La_W \La_{G+W}^{\dagger}) \leq T$.
\end{enumerate}
\end{definition}

We prove that for every patch, there exists a ``patch sparsifier'' supported on
$O(k)$ edges. Specifically, we prove the following theorem.

\begin{claim}\label{cl:patch}
Let $W=(V, E_W, \{w_e\}_{e\in E_W})$ be a $(k, T,\lambda^*)$-patch for $G$ with edge weights $w_e$
and $N \geq 8k$. Then there is a weighted graph $W_k = (V, E_{W_k}, \{\tilde w_e\}_{e\in E_{W_k}})$
with edge weights $\tilde w_e$ such that
\begin{enumerate}
\item $W_k$ has at most $N$ edges; $E_{W_k} \subseteq E_W$.
\item
$c_1 \min(N/T,1)  \lambda^* \La_{G+W} \preceq \La_{G+W_k} \preceq c_2 \La_{G+W}$,
for some absolute constants $c_1$ and $c_2$.
\item The total weight of edges, $\sum_{e\in E_{W_k}} \tilde w_k$, is at most $\min(1, N/T) \sum_{e\in E_W} w_e$.
\end{enumerate}
We say that $W_k$ is a patch sparsifier  of $W$ with respect to $G$.
\end{claim}

The claim will follow immediately from the following theorem, which is
is of independent interest. We will also show another (related) application
of this theorem in Section~\ref{sec:patch:algcon}.

\begin{theorem}
\label{thm:mainmatrix}
Suppose we are given a positive definite $n\times n$ matrix $X$ and a sequence of matrices $Y_i = v_i v_i^T$ ($i=1,\dots, m$) with
$$X + \sum_{i=1}^m Y_i = M^*,$$
and $\lambda_{\mathrm{max}}(M^*) \leq 1$.
Additionaly, suppose each matrix  $Y_i$ has cost $cost_i$ and $\sum_{i=1}^m cost_i = 1$.
Let $\lambda^* = \lambda_{k+1}(X)$, and $T= \lceil \tra (M^*-X)\rceil $. Then for every $N > 8k$ there exists a set of weights $w_i$ with $|\{w_i: w_i \neq 0\}| = N$ such that the matrix
$M = X + \sum_{i= 1}^m w_i Y_i$ satisfies,
$$
c_1\min(N/T,1) \cdot \lambda^* \cdot\lambda_{\mathrm{min}}(M^*)
\leq
\lambda_{\mathrm{min}}(M)
\leq
\lambda_{\mathrm{max}}(M)
\leq c_2,
$$
where $c_1$ and $c_2$ are some absolute constants,
and $\sum_{i=1}^m w_i cost_i \leq \min(1, N/T)$.
\end{theorem}
\textbf{Proof Overview.}
Our proof closely follows the approach of Batson, Spielman, and Srivastava~\cite{BSS}.
We construct matrix $M$ in $N$ steps; at each step we choose an index $i$ and weight $w_i$ and add $w_i Y_i$ to the sum $X + \sum_{i=1}^m w_i Y_i$.
Recall that Batson, Spielman, and Srivastava define two ``barriers'' $l$ and $u$ and maintain the property
that all eigenvalues of $M$ lie between $l$ and $u$. At each step, they increase $l$ and $u$ and update
matrix $M$ so that this property still holds. Finally, the ratio between $u$ and $l$ becomes very
close to $1$, which means that $\lambda_{\mathrm{min}}(M)$ is very close to $\lambda_{\mathrm{max}}(M)$.
During this process, they keep track not only of the smallest
and largest eigenvalues of $M$ but of all $n$ eigenvalues  to avoid accumulation of eigenvalues in
neighborhoods of $l$ and $u$. To this end, they define two potential functions, the
lower potential function $\Phi_l(M) = \sum_{i=1}^n \frac{1}{\lambda_i(M) - l}$ and the upper potential function
$\Phi^u(M) = \sum_{i=1}^n \frac{1}{u - \lambda_i(M)}$, and then ensure that $\Phi_l(M)$ and $\Phi^u(M)$
do not increase over time. That guarantees that all eigenvalues of $M$ stay far away from $l$ and $u$.

In our proof, however, we cannot keep an eye on all eigenvalues. After each step, only one eigenvalue increases,
and thus we need $\theta(n)$ steps to increase all eigenvalues participating in the definition of
$\Phi_l(M)$. But our goal is to ``patch'' $X$ in roughly $k$ steps. So we focus our attention only on $k$ smallest and $T$ largest eigenvalues.

Let $S$ be the eigenspace of $X$ corresponding to $k$ smallest eigenvalues, and $P_S$ be the projection onto $S$.
We define the lower potential function as follows,
$$\Phi_l(A) = \tra(P_S(A-lI)P_S)^{\dagger} = \sum_{i=1}^k \frac{1}{\lambda_i( \left.A\right|_S ) - l},$$
where $\left.A\right|_S$ denotes the restriction of $A$ to the space $S$ ($\left.A\right|_S$ is a $k\times k$ matrix).
Note that the space $S$ is fixed, and the eigenvector corresponding
to the smallest eigenvalue will not necessarily lie in $S$ after a few steps.
We want to ensure that after $N$ steps,
$$\sum_{i=1}^m w_i \bigl.Y_i\bigr|_S \succeq c \min(N/T,1) \sum_{i=1}^m \bigl.Y_i\bigr|_S = c \min(N/T,1) \bigl.(M^*-X)\bigr|_S,$$
or in other words,
$\lambda_{\mathrm{min}}(\left.(Z(\sum_{i=1}^m w_iY_i)Z)\right|_S) \geq c \min(N/T,1)$, where
$Z = {\left((P_S(M^* - X)P_S)^{\dagger}\right)}^{1/2}$.
To this end, we show how to update $M$ and $l$ so that $\Phi_l(Z(\sum_{i=1}^m w_i Y_i)Z)$ does not increase, and $l$
equals $c \min(N/T,1)$ after $N$ steps.
It remains to lower bound $\lambda_{\mathrm{min}}(M)$ in the entire space.
We know that all eigenvalues of $X$ (and therefore, of $M$)
 in $S^{\perp}$ are at least $\lambda^*$.
We show that that together with an upper bound on $\lambda_{\mathrm{max}}(M)$ implies
that
$\lambda_{\mathrm{min}}(M) \geq c_1\min(N/T,1) \cdot \lambda^*\lambda_{\mathrm{min}}(M^*)$
(the product of the lower bounds on $\lambda_{\mathrm{min}}$
in spaces $S$ and $S^{\perp}$ divided by the upper bound on $\lambda_{\mathrm{max}}$).

Similarly, we amend the definition of the upper potential function. Since we need to bound $\lambda_{\mathrm{max}}$ in the entire
space, we cannot restrict $\Phi^u(M)$ to a fixed subspace. For a matrix $A$, we consider the eigenspace of $A$ corresponding
to its largest $T$ eigenvalues. Denote it by $L_A(A)$; denote the projection onto $L(A)$ by $P_{L(A)}$. Then
$$
\Phi^u(A) = \tra(P_{L(A)}(uI-A)^{-1}P_{L(A)}) = \tra(P_{L(A)}(uI - A)P_{L(A)})^{\dagger}
= \sum_{i=n-T + 1}^N \frac{1}{u - \lambda_i(A)}.
$$
Note that both definitions of $\Phi^u(A)$ --- in terms of regular inverse and in terms of pseudoinverse --- are equivalent since $L(A)$ is an invariant subspace
of $A$. However, $\Phi_l(A)$ is not equal to $\tra(P_{S}(A - lI)^{-1}P_{S})$ in general since $S$ is not necessarily an invariant subspace of $A$.

Our algorithm and analysis are similar to those of Batson, Spielman, and Srivastava~\cite{BSS}. However, several complications arise because we are controlling eigenvalues in different
subspaces and, moreover, one of these subspaces, $L(A)$, is not fixed.

Let us summarize the proof.
We construct the matrix $M$ iteratively in $N$ steps.
Let $A^{(q)}$ be the matrix and $w_i^{(q)}$ be the weights after $q$ steps.
We define an auxiliary matrix $B^{(q)}$ as $Z(A^{(q)}-X)Z$. We have,
$$
A^{(q)} = X + \sum_i w_i^{(q)} Y_i;\quad
B^{(q)} = \sum_i w_i^{(q)} ZY_iZ = Z(A^{(q)}-X)Z.
$$
We will ensure that the following properties hold after each step
(for some values of constants $l_0$, $\delta_L$, $u_0$, $\delta_U$, $\epsilon_L$, $\epsilon_U$, which we will specify later).
\begin{enumerate}
\item $\Phi_{l_0}(B^{(0)}) \leq \epsilon_L$ and $\Phi^{u_0}(A^{(0)}) \leq \epsilon_U$.
\item Each matrix $A^{(q)}$ and $B^{(q)}$ is obtained by a rank-one update of the previous one:
$$
A^{(q+1)} =A^{(q)}+t Y_i,\quad
B^{(q+1)} =B^{(q)}+t ZY_iZ
$$
for some $i$.
\item Lower and upper potentials do not increase. Namely, for every $q =0,1,\dots,N$,
$$
\Phi^{u_0+(q+1)\delta_U}(A^{(q+1)}) \leq \Phi^{u_0 + q\delta_U}(A^{(q)})\leq \epsilon_U
\text{ and }
\Phi_{l+(q+1)\delta_L}(B^{(q+1)}) \leq \Phi_{l_0 + q\delta_L}(B^{(q)})\leq \epsilon_L.
$$

\item  At each step $q$,
$\lambda_{\mathrm{min}}(\bigl.B^{(q)}\bigr|_S) > l \equiv l_0 + q \delta_L$ and
$\lambda_{\mathrm{max}}(A^{(q)}) < u \equiv u_0 + q \delta_U$.
In particular, this condition ensures
that all terms in the definitions of upper and lower potentials are positive.
\item At each step $q$, the total cost is at at most $q/\max (N,T)$: $\sum w_i^{(q)} cost_i \leq q/\max (N,T)$.
\end{enumerate}

We present the complete proof in Sections~\ref{sec:barshift}
and \ref{sec:thmmainproof}.
In Section~\ref{sec:barshift}, we first find conditions under which we can update $A^{(q)}$ and $u$ (Lemma~\ref{upper_potential}), and $B^{(q)}$ and $l$ (Lemma~\ref{lower_potential}).
Then we show that both conditions can be simultaneously satisfied (Lemma~\ref{both_barriers}).
In Section~\ref{sec:basic}, we prove several theorems
that we need later to deal with a non-fixed subspace $L(A)$.
Finally, in Section~\ref{sec:thmmainproof}, we combine all pieces of the proof together.

\subsection{Some Basic Facts about Matrices}
\label{sec:basic}
\subsubsection{Sherman--Morrison Formula}
We use the Sherman--Morrison Formula, which describes the behavior of the inverse
of a matrix under rank-one updates. We first state the formula for regular inverse~\cite{GV96},
and then we show that a similar expression holds for the pseudoinverse.

\begin{lemma}[Sherman--Morrison Formula]\label{lem:path:sminv}
If $A$ is a nonsingular $n\times n$ matrix and $Y = vv^T$ is a rank-one update, then
\begin{equation}\nonumber
(A+Y)^{-1} =A^{-1} - \frac{A^{-1} Y A^{-1}}{1+A^{-1}\mprod Y}
\end{equation}
\end{lemma}

\begin{lemma}\label{lem:path:smpseudinv}
If $A$ is a symmetric (possibly singular) $n\times n$ matrix, $Y = vv^T$ is a rank-one update, then
$$
(A+PYP)^{\dagger} =A^{\dagger} - \frac{A^{\dagger} Y A^{\dagger}}{1+A^{\dagger}\mprod Y},
$$
where $P$ is the orthogonal projection on $\im(A)$.
\end{lemma}
\begin{proof}
Let $\bar v = Pv$ and $\bar Y = PYP = \bar v \bar v^T$ .
Note that $A^{\dagger} Y A^{\dagger} = A^{\dagger} \bar Y A^{\dagger}$, since $P A^{\dagger} = P$, and
$$A^{\dagger}\mprod \bar Y = \tra A^{\dagger} \bar Y = \tra A^{\dagger} (PYP) =
\tra (PA^{\dagger}P) Y = A^{\dagger}\mprod  Y.$$
We need to verify that
$$
(A+\bar Y)\left(A^{\dagger} - \frac{A^{\dagger}\bar Y A^{\dagger}}{1+A^{\dagger}\mprod \bar Y}\right) =\left(A^{\dagger} - \frac{A^{\dagger}\bar Y A^{\dagger}}{1+ A^{\dagger}\mprod
\bar Y}\right)(A+\bar Y)
= P.
$$
Since $A$ is a symmetric matrix, $AA^{\dagger} = A^{\dagger}A = P$. Since $P^2 = P$, $P\bar Y P = \bar Y$ and
$\bar Y A^{\dagger} \bar Y = \bar v \bar v^T A \bar v \bar v^T
= \bar v (A \mprod \bar Y) \bar v^T = (A \mprod \bar Y) \bar Y$.
We calculate,
\begin{align*}
(A+\bar Y)\left(A^{\dagger} - \frac{A^{\dagger} \bar Y A^{\dagger}}{1+ A^{\dagger} \mprod \bar Y}\right) &=
AA^\dagger+ \bar Y A^\dagger -\frac{\bar P}{Y A^\dagger+ \bar Y A^\dagger \bar Y A^\dagger}{1+ A^\dagger \mprod \bar Y}\\
&= P + \bar Y A^\dagger -\frac{(1+ A^\dagger\mprod \bar Y )\bar Y A^\dagger}{1
+ A^\dagger\mprod \bar Y}=P + \bar Y A^\dagger - \bar Y A^\dagger=P.
\end{align*}
Similarly,
$$(A^{\dagger} - \frac{A^{\dagger} \bar Y A^{\dagger}}{1+A^{\dagger} \mprod \bar Y})(A+\bar Y) = P.$$
\end{proof}
\subsubsection{Majorization}
\label{sec:basic}
\begin{lemma}\label{Majorization}(Majorization)
For every positive semidefinite matrix $A$, every projection matrix $P$, and every $r\in\{1,\dots,n\}$
\begin{equation}
\label{eq:majorization}
\sum_{i= n - r + 1}^n \lambda_i(A) \geq \sum_{i= n - r + 1}^n
\lambda_i(PAP).
\end{equation}
In particular, $\lambda_{\mathrm{max}}(A) \geq \lambda_{\mathrm{max}}(PAP)$.
\end{lemma}
\begin{proof}
Let $e_1, \dots, e_n$ be an orthonormal eigenbasis of $A$ so that $e_i$ has
eigenvalue $\lambda_i(A)$. Similarly, let $\tilde e_1, \dots, \tilde e_n$ be an orthonormal eigenbasis
of $PAP$ so that $\tilde e_i$ has eigenvalue $\lambda_i(PAP)$.
Write
$$\tilde e_i = \sum_{j=1}^n \langle e_j, \tilde e_i\rangle e_j.$$
Note that if $\lambda_i(PAP) \neq 0$ then
$\tilde e_i \in \im(P A P) \subseteq \mathrm{Im}(P)$ and $P \tilde e_i = \tilde e_i$.
Then
$$\lambda_i(PAP) = \tilde e_i^T PAP \tilde e_i
= \tilde e_i A \tilde e_i
= \sum_{j=1}^n \langle e_j, \tilde e_i\rangle^2 \lambda_j(A).$$
If $\lambda_i(PAP) = 0$ then trivially
$$\lambda_i(PAP) = 0 \leq \sum_{j=1}^n \langle e_j, \tilde e_i\rangle^2 \lambda_j(A).$$
Therefore,
$$\sum_{i=n - r + 1}^n \lambda_i(PAP) \leq
\sum_{i=n - r + 1}^n \sum_{j=1}^n \langle e_j, \tilde e_i\rangle^2 \lambda_j(A)=
\sum_{j=1}^n \left(\sum_{i=n - r + 1}^n \langle e_j, \tilde e_i\rangle^2\right) \lambda_j(A).
$$
That is, $\sum_{i=n - r + 1}^n \lambda_j(PAP)$ is at most
the sum of  $\lambda_j(A)$ with weights $\sum_{i=n - r + 1}^n \langle e_j, \tilde
e_i\rangle^2$.
The total weight of all $\lambda_1(A),\dots,\lambda_n(A)$ is $r$:
$$\sum_{i=n - r + 1}^n\underbrace{\sum_{j=1}^n \langle e_j, \tilde e_i\rangle^2}_{
\|\tilde e_i\|^2} =
\sum_{i=n - r + 1}^n \|\tilde e_i\|^2 = r.$$
The weight of each eigenvalue $\lambda_j(A)$ in the sum is at most $1$:
$$\sum_{i=n - r + 1}^n \langle e_j, \tilde e_i\rangle^2 \leq \sum_{i=1}^n \langle e_j,
\tilde e_i\rangle^2 = 1.$$
Therefore, the sum does not exceed the sum of the $r$ largest eigenvalues
$\sum_{i=n-r+1}^n \lambda_r(A)$.
\end{proof}

\begin{cor}\label{cor:Majorization}
For every positive semidefinite matrix $A$, every projection matrix $P$ and $u > \lambda_{\mathrm{max}}(A)$, the following inequality holds.
\begin{equation}
\Phi^{u}(PAP) =
\sum_{i=n-T + 1}^{n}\frac{1}{u-\lambda_i(PAP)}
\leq \sum_{i=n-T+1}^{n}\frac{1}{u-\lambda_i(A)} =\Phi^u(A)
\end{equation}
\end{cor}
\begin{proof}
The statement follows from the Karamata Majorization Inequality. The inequality
claims that for every two non-increasing sequences
that satisfy (\ref{eq:majorization}) and for every increasing convex function $f$,
$$
\sum_{i=n-k+1}^n f(\lambda_i(A))
\geq
\sum_{i=n-k+1}^n f(\lambda_i(PAP)).
$$
Plugging in $f(x) = \frac{1}{u-x}$ (defined on $(0,u)$), we obtain the desired inequality.
\end{proof}

\begin{lemma}\label{lem:patrial_trace}
Let $A$ be a  positive semidefinite matrix such that  $A \preceq I_n$. Assume $\mathrm{Tr}(A) \leq r \in \mathbb N$. Then for every positive semidefinite matrix $M$,
$
A\mprod M \leq \sum_{i=N-r+1}^{N} \lambda_i(M)
$.
\end{lemma}
\begin{proof}
By von Neumann's inequality~\cite{Mirsky},
$
A\mprod M = \tra(AM) \leq \sum_{i=1}^n \lambda_i(A) \lambda_i(M)
$.
Since $\sum_{i=1}^n \lambda_i(A) \leq r$ and all $\lambda_i(A) \leq 1$, we can easily see that the above product achieves its maximum when the largest $r$ eigenvalues of $A$ are $1$ and the rest are $0$. In this case, we have,
$
A\bullet M \leq \sum_{i=1}^n \lambda_i(A) \lambda_i(M) = \sum_{i=n-r+1}^n \lambda_i(M)
$.
\end{proof}

As a corollary we get the following result.
\begin{cor}\label{cor:traceineq}
Let $X$, $M^*$ and $T$ be as in Theorem~\ref{thm:mainmatrix}. Then
for any positive semidefinite matrix $U$, we have
$U\mprod (M^* - X) \leq \sum_{i=n-T+1}^{n} \lambda_i(U)$.
\end{cor}
\subsection{Barrier Shifts}\label{sec:barshift}
In this section, we analyze how we can update matrices $A^{(q)}$ and $B^{(q)}$, and
increment barriers $l$ and $r$ so that the upper and lower potentials do not increase.
Let us think of $\Phi^u(A)$ as a function of an $n^2$ dimensional vector (consisting of entries
of $A$). Then in the first approximation $\Phi^{u+\delta_U}(A + tY) \approx \Phi^{u+\delta_U}(A) + t Y\mprod U$, where $U$ is the gradient of $\Phi^{u+\delta_U}$ at $A$ ($U$ is an $n\times n$ matrix). Thus the potential function does not increase,
$\Phi^{u+\delta_U}(A + tY) \leq \Phi^u(A)$,  roughly when
$t Y\mprod \frac{U}{\Phi^{u}(A) - \Phi^{u+\delta_U}(A)} \leq 1$.
Similarly, $\Phi_{l+\delta_L}(B + tY) \leq \Phi_l(B)$, roughly when
$tY\mprod \frac{L}{\Phi_{l+\delta_L}(B) - \Phi_{l}(B)}\geq 1$, where $L$ is the gradient
of $\Phi_{l+\delta_L}$
at $B$. Following~\cite{BSS}, we make these statements precise (we need to
take into account lower order terms).
We define matrices $U_A$ and $L_B$,
\begin{align*}
U_A &= \frac{((u+\delta_U)I-A)^{-2}}{\Phi^u(A) -\Phi^{u+\delta_U}(A)}+((u+\delta_U)I-A)^{-1};\\
L_B &= \frac{(P_{S}(B-(l+\delta_L)I)P_{S})^{\dagger 2}}{\Phi_{l+\delta_L}(B) -\Phi_{l}(B)}-(P_{S}(B-(l+\delta_L)I)P_{S})^{\dagger}
\end{align*}
\begin{lemma}\label{upper_potential} (Upper Barrier Shift)
Suppose $\lambda_{\text{max}}(A)<u$ and $Y = v v^T$ is a rank-one update. If
$
U_A \mprod Y \leq \frac{1}{t}
$
then
$
\Phi^{u+\delta_U}(A+tY) \leq \Phi^u(A)
$
and
$
\lambda_{\mathrm{max}}(A+tY)<u+\delta_U
$.
\end{lemma}
\begin{proof}
Let $u' = u+\delta_U$ and $P = P_{L(A+tY)}$.  By the Sherman--Morrison formula (Lemma (\ref{lem:path:sminv})), we can write the updated potential as:
\begin{align*}
\Phi^{u+\delta_U}(A+tY)  &=
\tra P(u'I-A-t Y)^{-1}P
=\tra P\left((u'I-A)^{-1}+\frac{t(u'I-A)^{-1}Y(u'I-A)^{-1}}{1-t(u'I-A)^{-1}\mprod Y}\right)P\\
&=\tra P(u'I-A)^{-1}P+\tra \frac{tP(u'I-A)^{-1}Y(u'I-A)^{-1}P}{1-t(u'I-A)^{-1}\mprod Y}\\
&\leq \Phi^{u+\delta_U}(PAP) + \frac{t(u'I-A)^{-2}\mprod Y}{1-t(u'I-A)^{-1}\mprod Y}\\
&\leq\Phi^{u+\delta_U}(A)+\frac{t(u'I-A)^{-2}\mprod Y}{1-t(u'I-A)^{-1}\mprod Y}\\
&=\Phi^u(A) -(\Phi^u(A) -\Phi^{u+\delta_U}(A))+
\frac{(u'I-A)^{-2}\mprod Y}{1/t-(u'I-A)^{-1}\mprod Y}
\end{align*}
Here, we used Corollary~\ref{cor:Majorization} for the inequality on line 4.

Substituting $U_A \mprod Y \leq 1/t$ gives $\Phi^{u+\delta_U}(A+tY) \leq \Phi^u(A)$.
The statement about $\lambda_{\mathrm{max}}$ follows from continuity of eigenvalues.
\end{proof}
\begin{lemma}\label{lower_potential} (Lower Barrier Shift)
Suppose $\lambda_{\mathrm{min}}(\left.B\right|_S)>l+\delta_L$ and $Y = vv^T$ is a rank-one update. If
$
L_B \mprod Y \geq 1/t
$
then
$
\Phi_{l+\delta_L}(B+tY) \leq \Phi_l(B)
$
and
$
\lambda_{\mathrm{min}}(\left.(B+tY)\right|_S)>l+\delta_L
$.
\end{lemma}
\begin{proof}
We proceed as in the proof for the upper potential. Let $l' = l + \delta_L$ and $P=P_S$.
By the Sherman--Morrison formula for the pseudoinverse (Lemma~\ref{lem:path:smpseudinv}),
we have:
\begin{align*}
\Phi_{l+\delta_L}(B+tY)  &= \tra(P(B + tY-l'I)P)^{\dagger}= \tra(P(B-l'I)P+tPYP)^{\dagger}\\
&= \tra(P(B-l'I)P)^{\dagger}-\frac{t\tra((P(B-l'I)P)^\dagger Y (P(B-l'I)P)^\dagger)}{1+t(P(B-l'I)P)^{\dagger}\mprod  Y}\\
&=\Phi_l(B)+(\Phi_{l+\delta_L}(B)-\Phi_l(B))-\frac{t(P(B-l'I)P)^{\dagger 2}\mprod
Y}{1+t(P(B-l'I)P)^{\dagger}\mprod Y}
\end{align*}
Note that matrix  $U_A$ is positive semidefinite.
Rearranging shows that $\Phi_{l+\delta_L}(B + Y)  \leq \Phi_l(B)$ when $ L_A(\pi) \geq 1/t$. It is immediate that $\lambda_{\text{min}}(P_S(A+t\pi\pi^T)P_S)>l+\delta_L$ since $\lambda_{\text{min}}(P_SAP_S)>l+\delta_L$.
\end{proof}

Now we prove that we can choose $Y_i$ and $t$ so that conditions of both lemmas are satisfied.
\begin{lemma}\label{both_barriers}(Both Barriers)
If $\Phi^{u}(A) \leq \epsilon_U$ and $\Phi_{l}(B) \leq \epsilon_L$ and $\epsilon_U, \epsilon_L,\delta_U,\delta_L$ satisfy
$$
0\leq \frac{1}{\delta_U}+\epsilon_U + \max(N, T) \leq \frac{1}{\delta_L} - \epsilon_L,
$$
and $X$, $Y_i$, $cost_i$, $Z$, $T$ and $N$ as in Theorem~\ref{thm:mainmatrix},
$M^*-X$ is non-singular on $S$,
then there exists $i$ and positive $t$ for which
\begin{align}
L_B\mprod (ZY_iZ) &\geq 1/t \geq U_A \mprod Y_i, \text{ and} \label{bothOne}\\
cost_i \cdot t &\leq 1/\max(N,T).\label{bothTwo}
\end{align}
\end{lemma}

We will use the following lemma
\begin{lemma} \label{boundsonUandL}
$
\sum_{i=1}^m U_A \mprod Y_i  \leq  \frac{1}{\delta_U} + \epsilon_U
$
and
$
\sum_{i=1}^m L_B \mprod (Z Y_i Z) \geq  \frac{1}{\delta_L} - \epsilon_L.
$
\end{lemma}

\begin{proof}
1.
We use Corollary~\ref{cor:traceineq} to bound the Frobenius product of $Y_i$ with each of the two summands
in the definition of $U_A$ (note that they are positive semidefinite), we get
\begin{align*}
\sum_{i=1}^m U_A\mprod Y_i &= U_A \mprod \sum_{i=1}^m Y_i = U_A \mprod (M^*-X)\\
&=
\frac{((u+\delta_U)I-A)^{-2}}{{\Phi^u(A) -\Phi^{u+\delta_U}(A)}}\mprod (M^*-X)
+((u+\delta_U)I-A)^{-1}\mprod (M^*-X) \\
&\leq \sum_{i=n-T+1}^n \lambda_i\left(\frac{((u+\delta_U)I-A)^{-2}}{{\Phi^u(A) -\Phi^{u+\delta_U}(A)}}\right) +
\sum_{i=n-T+1}^n \lambda_i\left(((u+\delta_U)I-A)^{-1}\right)\\
&=
\frac{
\sum_{i=n-T+1}^n
\frac{1}{(u + \delta_U - \lambda_i(A))^2}}{{\Phi^u(A) -\Phi^{u+\delta_U}(A)}}
 +
\sum_{i=n-T+1}^n \frac{1}{(u+\delta_U)-\lambda_i(A)}
\end{align*}
Note that the first term is at most $1/\delta_U$, since
\begin{align*}
\sum_{i=n-T+1}^n
\frac{1}{(u + \delta_U - \lambda_i(A))^2} &\leq
\sum_{i=n-T+1}^n
\frac{1}{(u - \lambda_i(A))(u + \delta_U - \lambda_i(A))}\\
&=
\frac{1}{\delta_U}\sum_{i=n-T+1}^n
\left(\frac{1}{u  - \lambda_i(A)} -  \frac{1}{(u + \delta_U) - \lambda_i(A)}\right)
= \frac{{\Phi^u(A) -\Phi^{u+\delta_U}(A)}}{\delta_U}
\end{align*}
and the second term equals $\Phi^{u+\delta_U}(A)$. Thus
$\sum_{i=1}^m U_A\mprod Y_i
\leq \epsilon_U + 1/\delta_U$.

\noindent
2. Let $P$ be the projection on $\im(M^*-X)$. Since
$(M^*-X)$ is non-singular on $S$, $P P_S = P_S$. We have,
\begin{align*}
\sum_{i=1}^m L_B\mprod ZY_iZ &= L_B \mprod \sum_{i=1}^m ZY_iZ  = L_B \mprod Z(M^*-X)Z = L_B \mprod P
\\
&=\tra\left(\frac{(P_S(B-(l+\delta_L)I)P_S)^{\dagger 2}}{\Phi_{l+\delta_L}(B) -\Phi_{l}(B)}-(P_{S}(B-(l+\delta_L)I)P_{S})^{\dagger}\right)\\
\\
&= \frac{\sum_{i=1}^{k}\frac{1}{(\lambda_i(\left.B\right|_S)-(l+\delta_L))^2}}{
\Phi_{l+\delta_L}(B) -\Phi_{l}(B)}-\sum_{i=1}^{k}\frac{1}{\lambda_i(\left. B \right|_S)-(l+\delta_L)}\\
&\geq 1/\delta_L - \epsilon_L,
\end{align*}
where the last line follows from Claim~3.6 in \cite{BSS}.

\end{proof}

\begin{proof}(Of Lemma~\ref{both_barriers})
For the previous lemma, we get:
$\sum_{i=1}^m (U_A \mprod Y_i + \max(N,T)cost_i) \leq \frac{1}{\delta_U} + \epsilon_U + \max(N,T) \leq
L_B\mprod (ZY_iZ).$
Thus for some $i$, $U_A \mprod Y_i + \max(N,T) cost_i \leq L_B\mprod (ZY_iZ)$. Letting $t = (L_B\mprod (ZY_iZ))^{-1}$, we
satisfy (\ref{bothOne}) and (\ref{bothTwo}).
\end{proof}
\subsection{Proof of Theorem~\ref{thm:mainmatrix}}
\label{sec:thmmainproof}
Now we are ready to prove Theorem~\ref{thm:mainmatrix}. We assume that $M^*-X$
is non-singular on $S$ (which we can ensure by an arbitrary small pertrubation).

We start with $A^{(0)} = X$, $B^{(0)} = 0$ and all weights $w^{(0)}_i = 0$.
We define parameters as follows,
\begin{align*}
\delta_L &= 1/(2\max(N,T)),&
\epsilon_L &= 1/(4\delta_L),&
l_0 &= -4k\delta_L,\\
\delta_U &=4\delta_L,&
\epsilon_U &= 1/(4\delta_L),&
u_0 &= 4T\delta_L + 1,
\end{align*}
so as to satisfy conditions of Lemma~\ref{both_barriers},
$\Phi^u(A^{(0)}) = \Phi^u(X) = \sum_{i=1}^T \frac{1}{u_0 - \lambda_{n+1-i}(X)} \leq T/(u_0-1) =\epsilon_U$,
$\Phi_l (B^{(0)})= \sum_{i=1}^k \frac{1}{0 - l_0} = - k/l_0 = \epsilon_L$,
$1/\delta_U + \epsilon_U + \max(N,T) = \frac{3}{2} \max(N,T)= 1/\delta_L - \epsilon_L$.
Then we iteratively apply Lemma~\ref{both_barriers}. At iteration $q$, we
find an index $i$ and a positive $t$ such that
$L_{B^{(q)}}(ZY_iZ) \geq 1/t \geq U_{A^{(q)}}(Y_i)$, $cost_i \cdot t \leq 1/\max (N,T)$,
and increment the weight of matrix $Y_i$ by $t$: $w^{(q+1)}_i = w^{(q)}_i + t$; update $l = l + \delta_L$ and $u = u + \delta_U$. The total cost increases by at most $1/\max(N,T)$.
Finally, after $N$ iterations we obtain matrices $A^{(N)}$ and $B^{(N)}$ with
\begin{align*}
\lambda_{\mathrm{max}}(A^{(N)}) &\leq u_0 + N \delta_U = 2(N + T)/\max(N,T) + 1\equiv \theta_{\mathrm{max}}\\
\lambda_{\mathrm{min}}(\bigl.B^{(N)}\bigr|_S) &\geq l_0 + N \delta_L =
(N/2 - 2k)/\max(N,T) \equiv \theta_{\mathrm{min}}.
\end{align*}
Now consider an arbitrary unit vector $v$. Let $v = v_S + v_{S^\perp}$, where $v_S \in S$
and $v_{S^\perp} \perp S$.
Since $B^{(N)} \succeq \theta_{\mathrm{min}} P_S$ and $v_S\in S$,
\begin{align*}
v_S^TA^{(N)}v_S &= v_S^T(X + (P_S(M^* - X)P_S)^{1/2} B^{(N)}
(P_S(M^* - X)P_S)^{1/2}) v_S \\
&\geq v_S^T(X + (P_S(M^* - X)P_S)^{1/2}   \theta_{\mathrm{min}} P_S
(P_S(M^* - X)P_S)^{1/2}) v_S \\
&=
\theta_{\mathrm{min}} v_{S}^T  M^* v_{S} +
(1 - \theta_{\mathrm{min}}) v_{S}^T  X v_{S} \geq \theta_{\mathrm{min}}
\lambda_{\mathrm{min}}(M^*) \|v_{S}\|^2.
\end{align*}
On the other hand,
$v_{S^\perp}^T A^{(N)} v_{S^\perp} \leq \theta_{\mathrm{max}} \|v_{S\perp}\|$.
Thus from the triangle inequality for the norm induced by $A^{(N)}$, we get
$$(v^T A^{(N)} v)^{1/2} \geq  \theta_{\mathrm{min}}^{1/2}
\lambda_{\mathrm{min}}(M)^{1/2}\|v_{S}\| -
\theta_{\mathrm{max}}^{1/2} \|v_{S^\perp}\| \geq \theta_{\mathrm{min}}^{1/2}\lambda_{\mathrm{min}}(M)^{1/2} -
 (\theta_{\mathrm{max}}^{1/2} + \theta_{\mathrm{min}}^{1/2} \lambda_{\mathrm{min}}(M)^{1/2}) \|v_{S^\perp}\|.$$
On the other hand, since $S$ is an eigenspace of $X$ corresponding to $k$ smallest
eigenvalues,
$$(v^T A^{(N)} v)^{1/2}\geq (v^T X v)^{1/2} \geq (v_{S^\perp}^T X v_{S^\perp})^{1/2}
\geq {\lambda^*}^{1/2} \|v_{S^\perp}\|.
$$
One of the two bounds above for $(v^T A^{(N)} v)^{1/2}$ increases and the other decreases as $\|v_{S^\perp}\|$ increases. They are equal when
$\|v_{S^\perp}\| =
\frac{\theta_{\mathrm{min}}^{1/2}\lambda_{\mathrm{min}}(M^*)^{1/2}}{{\lambda^*}^{1/2}
 + \theta_{\mathrm{max}}^{1/2} + \theta_{\mathrm{min}}^{1/2}\lambda_{\mathrm{min}}(M^*)^{1/2}}$.
Therefore,
$(v^T A^{(N)} v)^{1/2} \geq
\frac{\theta_{\mathrm{min}}^{1/2}{\lambda^*}^{1/2}
\lambda_{\mathrm{min}}(M^*)^{1/2}}{{\lambda^*}^{1/2} + \theta_{\mathrm{max}}^{1/2} +
\theta_{\mathrm{min}}^{1/2} \lambda_{\mathrm{min}}(M^*)^{1/2}}$.
We conclude that
$$
\lambda_{\mathrm{min}}(A^{(N)}) = \min_{v:\|v\|=1} v^T A^{(N)} v \geq
\frac{\theta_{\mathrm{min}}\lambda^*\lambda_{\mathrm{min}}(M^*)}{
\left({\lambda^*}^{1/2}
 + \theta_{\mathrm{max}}^{1/2} + \theta_{\mathrm{min}}^{1/2}
 \lambda_{\mathrm{min}}(M^*)^{1/2}\right)^2}
.$$
Plugging in the values of parameters, we get the statement of the theorem for $M= A^{(N)}$.
The total cost is at most $N/\max(N, T) = \min(1, N/T)$.
\hfill\qed

\vspace{0.2cm}
\noindent
Finally, we prove Claim~\ref{cl:patch}.
\begin{proof}[Claim~\ref{cl:patch}]
Let $V = \im(\La_{G+W}) = \ker(\La_{G+W})^{\perp}$. Let $\La_e$ be the Laplacian of
the edge $e$. Define
\begin{align*}
X &= \Bigl.\left((\La_{G+W}^{\dagger})^{1/2}\La_G (\La_{G+W}^{\dagger})^{1/2}\right)\Bigl|_V,\\
Y_e &= w_e \Bigl.\left((\La_{G+W}^{\dagger})^{1/2}\La_e (\La_{G+W}^{\dagger})^{1/2}\right)\Bigl|_V,\\
cost_e &= w_e/\left(\sum\nolimits_{d\in E_W} w_d\right).
\end{align*}
Since $\La_G + \sum_{e\in E_W}^m w_e \La_e = \La_{G+W}$, we have $X + \sum_{e\in E_W} Y_e = I$.
By the definition of the $(k,T,\lambda^*)$-patch, $\tra(I-X) \leq T$ and
$\lambda^* \leq \lambda^{k+1}(X) $. We apply Theorem~\ref{thm:mainmatrix} to matrices $X$,
$Y_e$ and $M^* = I$. We obtain a set of weights $\rho_e$ --- supported on at most $N$
edges --- such that
$$
c_1\min(N/T,1) \cdot \lambda^*
\leq
\lambda_{\mathrm{min}}\left(X+ \sum\nolimits_{e\in E_W} \rho_e Y_e\right)
\leq
\lambda_{\mathrm{max}}\left(X+ \sum\nolimits_{e\in E_W} \rho_e Y_e\right)
\leq c_2,
$$
Let $\tilde w_e = \rho_e w_e$. Weights $\tilde w_i$ define subgraph $W_k$ with
at most $N$ edges. It follows that
$$c_1 \min(N/T,1)  \lambda^* \La_{G+W} \preceq \La_{G+W_k} \preceq c_2 \La_{G+W}.$$
The total weight of edges of $W_k$ is $\sum_{e\in E_W} \rho_e w_e =
(\sum_{e\in E_W} \rho_e cost_e) \sum_{d\in E_W} w_d \leq \min (1, N/T) \sum_{d\in E_W} w_d$.
\end{proof}

\section{Constructing Nearly-Optimal Ultrasparsifiers}
\label{sec:patch:ultra}

We now apply our subgraph sparsification to build ultrasparsifiers.
Recall that a weighted graph $U$ is a $(\kappa,k)$-ultrasparsifier
  of another graph  $G$ if  $U \preceq G \preceq \kappa \cdot U$ and
  $U$ has only $n-1+k$ edges, where $n$ is the number of vertices in $U$ and $G$.
The main result of this section is the following theorem.

\begin{theorem}\label{thm:patch:ultra}
For any integer $k>0$, every graph has an  $(\frac{n}{k}\log n\, \tilde{O}(\log \log n),k)$--ultrasparsifier.
\end{theorem}

Our basic idea to build a good ultrasparsifier $U$ is quite simple.
Without loss of generality, we can assume that $G$ is connected
  and has $O(n)$ edges.
Otherwise given a graph $G$,  we can first find a linear size
  sparsifier using \cite{BSS}, for each of its connected components,
  and build a good ultrasparsifier for each component.
Because $U$ is only $k$ edges aways from a tree,
  our construction starts with good tree $T$.
As it will be much more clear below, the quality of a tree
  is measured by its {\em stretch}, as introduced
  by Alon, Karp, Peleg and West \cite{AKPW}.

Suppose $T$ is a spanning tree of $G = (V,E,w)$.
For any edge $e\in E$,  let $e_1,\cdots, e_k \in F$
  be the edges on the unique path in $T$ connecting
  the endpoints of $e$.
The \textit{stretch} of $e$ w.r.t. $T$ is given
  by $\text{st}_T(e) =w(e)(\sum_{i=1}^k \frac{1}{w(e_i)})$.
The stretch of the graph $G$ with respect to $T$ is
  defined by $\text{st}_T(G) =\sum_{e\in E} \text{st}_T(e).$
Our construction will start with a spanning tree with the lowest
    possible stretch.
  By \cite{ABN08}, we can in polynomial time grow a spanning tree $T$ with
$$\text{st}_T(G) = O(n \log n \log \log n (\log \log \log n)^3).$$

\begin{remark}For the sake of simplicity of the presentation, we will show the construction of ultrasparsifiers with $\Theta(k)$ edges. We note that by choosing the appropriate constants, the number of edges can be made exactly $k$.
\end{remark}
Let $\kappa = c_1\cdot\text{st}_T(G)/k$  for a sufficiently large constant $c_1$.
Our job is to choose $\Theta(k)$ more weighted
  edges $\tilde{W}$ and set $U = T+\tilde{W}$ such that
  $ c_2 \cdot U \preceq G \preceq \kappa \cdot U$,
   for a constant $c_2$.
To this end, let $W = (1/(c_3\kappa)) \cdot G$, for some constant $c_3$.
Then, $G = c_3 \kappa \cdot W  \preceq c_3 \kappa \cdot (W+T) $.
Also, because  $T \preceq G$, we have
  $T+W \preceq (1+1/(c_3\kappa))G \preceq c_4 \cdot G$, for a constant
  $c_4$.
Therefore, if we can find a $\Theta(k)$--edge subgraph $\tilde{W}$ of
  $W$ such that $T +\tilde{W} \preceq  \Theta(1) \cdot (T+W)$, we can
  then build a $n-1 + \Theta(k)$ edge graph $U = T+\tilde{W}$
  satisfying
  $ c_2 \cdot U \preceq G \preceq \kappa \cdot U$
   (if we choose our constants $c_i$'s carefully).

To apply our subgraph sparsification results to construct $\tilde{W}$,
  we  use the following structure result of Spielman and Woo
  (\cite{SW}:  Theorem 2.1 and Corollary 2.2).

\begin{lemma}\label{lem:patch:SW1} (Theorem 2.1 in \cite{SW})
(1) $\text{Tr}({\La^\dagger_T}^{1/2}\La_G{\La^\dagger_T}^{1/2}) =\text{st}_T(G).$
(2)
For every $t>0$, the number of eigenvalues of
  ${\La^\dagger_T}^{1/2}\La_G{\La^\dagger_T}^{1/2}$ greater than $t$ is at most $\text{st}_T(G)/t$.
\end{lemma}

We now use Lemma~\ref{lem:patch:SW1} to prove the following lemma,
from which Theorem~\ref{thm:patch:ultra} follows directly.
\begin{lemma}\label{lem:ultra:patch}
   $W$ is a $(k,O(k),\Theta(1))$--patch for $T$.
\end{lemma}
\begin{proof}
Let
$\lambda_i = \lambda_{i}((\La_{T+W}^{\dagger})^{1/2}\La_T(\La_{T+W}^{\dagger})^{1/2})$ be the $i$-th eigenvalue,
and $y_i$ be the corresponding eigenvector. Let $x_i = L_{T+W}^{1/2}y_i$. Then,
\begin{eqnarray*}
\lambda_i = \lambda_{i}((\La_{T+W}^{\dagger})^{1/2}\La_T(\La_{T+W}^{\dagger})^{1/2})
 = \frac{x_i^T \La_T x_i}{x_i^T \La_T x_i+  x_i^T \La_W x_i}
 =  \frac{x_i^T \La_T x_i}{x_i^T \La_T x_i+  x_i^T \La_G x_i/(c_3\kappa)},
\end{eqnarray*}
implying
\begin{eqnarray*}
 \frac{x_i^T \La_G x_i}{x_i^T \La_T x_i}  =  \frac{1-\lambda_i}{\lambda_i}
 c_3\kappa = \left(\frac{1-\lambda_i}{\lambda_i}\right) c_3 c_1 \frac{\text{st}_T(G)}{k}
= \frac{\text{st}_T(G)}{\frac{k}{c_1c_3}\frac{\lambda_i}{1-\lambda_i}}
\end{eqnarray*}
It follows from the definition of $\lambda_i$ that $0\leq \lambda_i <
1$.
Hence, $(1-\lambda_{i-1})/\lambda_{i-1} \geq (1-\lambda_{i})/\lambda_{i} $.
By Courant---Fischer theorem and the property 2 of Lemma \ref{lem:patch:SW1},
 we have $k \leq
 \frac{k}{c_1c_3}\frac{\lambda_{k+1}}{1-\lambda_{k+1}}.$
Therefore, $\lambda_{k+1} \geq \frac{c_1c_3}{1+c_1c_3} = \Theta(1)$. We also have,

\begin{eqnarray*}
\tra\left((\La^{\dagger}_{T+W})^{1/2}\La_W{(\La^{\dagger}_{T+W}})^{1/2}\right)
& \leq &
 \tra\left((\La^{\dagger}_{T})^{1/2}\La_W{(\La^{\dagger}_{T}})^{1/2}\right)
= \frac{1}{c_3\kappa}\tra\left((\La^{\dagger}_{T})^{1/2}\La_G{(\La^{\dagger}_{T}})^{1/2}\right)\\
 & \leq & \frac{k}{c_3c_1\text{st}_T(G)} \text{st}_T(G) =
 \frac{k}{c_3c_1} = \Theta(k).
\end{eqnarray*}

We proved that $W$ is a $(k,O(k),\Theta(1))$--patch for $T$.
\end{proof}

We next show that the parameters of the ultrasparsifiers we obtained
are optimal, up to low order terms.

\begin{theorem} Let $G$ be a Ramanujan $d$-regular expander graph, for
  some constant $d$. Let $U$ a $(\kappa,N)$ ultrasparsifier for $G$. Then
$
\kappa \geq \frac{n}{N} \log n.
$
\label{thm:optimsparsifiers}
\end{theorem}

\begin{proof} Let T be a low-stretch spanning tree of $G$, as above. As mentioned in \cite{ABN08}, $\text{st}_T(G) =\Omega( m\log n)$ where $m$ is the number of edges of the original graph.
From lemma \ref{lem:patch:SW1}, and the conditions on the stretch of $T$ we have
$
\text{Tr}(\La_G{\La_T}^\dagger) =\text{st}_T(G) \geq C\cdot n \log n
$
for some constant $C$.

Since $x^T\La_G x = \Theta(1)$ for the expander, the above inequality implies that $\sum_{i =1} ^n \frac{1}{x^T \La_T x} \geq n \log n$
where $x_i$ are the eigenvectors of $\La_G(\La_T)^\dagger$. It is immediate from Markov's inequality that there exists some $k$ such that ${x_k}^T \La_T x_k \leq \frac{C_1k}{n \log n}$. Assume that for all $i \leq k$ we have ${x_i}^T \La_T x_i \leq {x_k}^T \La_T x_k \leq \frac{C_1k}{n \log n} $. (Otherwise take $k' < k$ appropriately). Then also $\lambda_k(\La_T) \leq \frac{C_1k}{n \log n}$. By the minmax theorem for eigenvalues this implies that adding $N = k-2$ edges to $T$ will result to a graph $U$ with $\lambda_2(\La_U) \leq \lambda_k(\La_T)\leq \frac{C_1k}{n \log n}$. Thus any ultrasparsifier $U$ with $N$ edges will have
\begin{eqnarray*}
 C_2 = \lambda_2(\La_G) \leq \kappa \lambda_2(\La_U) \leq \frac{C_1k}{n \log n}
\Rightarrow \kappa \geq \Omega(\frac{n\log n}{k}) =\Omega(\frac{n\log n}{N})
\end{eqnarray*}
\end{proof}

\section{Maximizing Algebraic Connectivity by Adding few edges}
\label{sec:patch:algcon}

In this section, we present an approximation algorithm for the
following problem: given a graph $G = (V, E_{base})$, a set of
candidate edges $E_{cand}$, and a parameter $k$, add at most $k$
candidate edges to $G$ so as to maximize its algebraic connectivity,
that is, find a subset $E\subset E_{cand}$ that maximizes
$\lambda_2(\La_{G + E})$. The problem was introduced by Ghosh and
Boyd~\cite{BG}, who presented a heuristic for it. It is known that
the problem is NP-hard \cite{damon}. But prior to this work, no
approximation algorithm was known for it. 

We use two upper bounds for  the cost of the combinatorial solution
in order to prove an approximation guarantee: one upper bound is the
SDP value, $\lambda_{SDP}$, and the other is $\lambda_{k+2}(\La_G)$
(see Lemma~\ref{lem:boundOnLambdaK}). Note that neither of these two
bounds are good approximations for the value of the optimum solution
by themselves (for instance, if $G$ consists of $n$ isolated
vertices, $(V, E_{cand})$ is an expander, $k < n$, then the value of
the combinatorial solution is $0$ but $\lambda_{SDP} \sim k/n$), but
their combinations lead to a good upper bound for the optimum
solution $\lambda_{OPT}$.

For clarity and simplicity of exposition, we assume here that $(V,
E_{base})$ and $(V, E_{cand})$ are bounded degree graphs with the
maximum degree $\Delta$.
Our algorithm uses a natural semidefinite relaxation that was also
used by Ghosh and Boyd~\cite{BG}. We introduce a variable $w_e$ (the
weight of the edge $e$) for each candidate edge $e\in E_{cand}$; add
constraints that all edge weights are between $0$ and $1$, and the
total weight is at most $k$. Then we require that $\lambda_2(\La_G +
\sum_e w_e\La_e) \geq \lambda_{SDP}$ (where $\La_e$ is the Laplacian
of the edge $e$). We do that by adding an SDP constraint $\La_G +
\sum_e w_e \La_e \succeq \lambda_{SDP} P_{(1,\dots,1)^{\perp}}$,
where $P_{(1,\dots,1)^{\perp}}$ is the projection on the space
orthogonal to $(1,\dots,1)^{\perp}$.
We get the following SDP relaxation.
\begin{align*}
\text{maximize: } & \lambda_{SDP},\\
\text{subject to: }& \La_G + \sum_{e\in E_{cand}} w_e \La_e \succeq \lambda_{SDP}
\cdot P_{(1,\dots,1)^{\perp}},\\
&\sum_{e\in E_{cand}} w_e \leq k,\\
&0\leq w_e  \leq 1 \text{ for every } e\in E_{cand}.
\end{align*}

We solve the semidefinite program and obtain solution $\{w_e\}_{e\in
E_{cand}}$. The total weight of all edges is $k$, however, the
number of edges involved, or the support of the solution could be
significantly higher than $k$.

We use our algorithm to \textit{sparsify} the SDP solution using
Theorem~\ref{thm:mainmatrix}. More precisely, we apply
Theorem~\ref{thm:mainmatrix} with $X = \La_G/(4\Delta)$ and $Y_e =
w_i \La_e/(4\Delta)$ restricted to the space $(1,\dots, 1)^{\perp}$,
$N = 8k$, $T = \tra(\sum_e w_e \La_e)/(4\Delta) \leq k$ and $cost_i =
w_i$ (we divide $\La_G$ and $\La_e$ by $4\Delta$ to ensure that
$\lambda_{\mathrm{max}}(X + \sum_e Y_i) \leq 1$). We get a set of
weights $\rho_e$ supported on at most $8k$ edges s.t.
$$
\frac{1}{4\Delta}
\lambda_{\mathrm{2}} (\La_G + \sum_e \rho_e w_e \La_e)
= \lambda_{\mathrm{min}} (X + \sum_e \rho_e Y_e)
\geq  c \lambda_{k+2}(X) \lambda_{\mathrm{min}}(X + \sum_e Y_e)
\geq c \frac{1}{(4\Delta)^2} \lambda_{k+2}(\La_G) \lambda_{SDP}.
$$
That is, we obtain a combinatorial weighted solution $\tilde w_e = \rho_i w_i$
whose value is at least $c\lambda_{k+2}(\La_G)\lambda_{SDP}/(4\Delta)$
(if $k+2 > n$, the value is at least $c\lambda_{SDP}$).
We next show that
$\lambda_{SDP} \geq \lambda_{OPT}$ and $\lambda_{k+2} (G) \geq \lambda_{OPT}$.
Therefore, the value of the solution is at least $c\lambda_{OPT}^2/\Delta$.

\begin{lemma}\label{lem:boundOnLambdaK}
The value of the optimal solution, $\lambda_{OPT}$, is at most $\lambda_{k+2}(\La_G)$.
\end{lemma}
\begin{proof}
Consider the optimal solution $E$. Let $\La_E$ be the Laplacian of
the graph formed by $E$. Note that $\mathrm{rank}(\La_E) \leq |E| \leq k$, therefore,
$\dim \ker \La_E \geq n - k$.
Let $S$ be the $k+1$-dimensional space spanned by the eigenvectors of
$\La_G$ corresponding to $\lambda_2(\La_G), \dots, \lambda_{k+2}(\La_G)$. Since
$\dim S+ \dim \ker E > n$, spaces $S$ and $\ker \La_E$ have a non-trivial intersection.
Choose a unit vector $v\in \ker S\cap \La_E$. We have
$v(\La_{G}+\La_E)v^T \leq \lambda_{k+2}(\La_G) + 0 =\lambda_{k+2}(\La_G)$. Also $v$ is
orthogonal to the vector $(1,\dots,1)^\perp$. Therefore,
$\lambda_{OPT} = \lambda_2(\La_{G}+\La_E) \leq \lambda_{k+2}(\La_G)$.
\end{proof}

The edges in the support of $\tilde w_e$, $E = \{\tilde w_e: \tilde w_e \neq 0\}$,
form a non-weighted combinatorial solution. Since
$\lambda_{\mathrm{max}}(\La_X + \sum_e \tilde w_e \La_e) = O(\Delta)$,
all weights $\tilde w_e$ are bounded by $O(\Delta)$, and thus
the algebraic connectivity of $G+E$ is at least
$c\lambda_{k+2}(\La_G)\lambda_{SDP}/\Delta^2$.
\begin{theorem}\label{thm:algconn}
There is a polynomial time approximation algorithm that finds
a solution of value at least $c\lambda_{OPT}^2/\Delta$ supported on at most
$8k$ edges with total weight at most $k$. If $k\geq n$ the algorithm
finds a constant factor approximation.
\end{theorem}
We present two corollaries for special instances of the problem.
\begin{cor}
If it is possible to make $G$ an expander  by adding $k$ edges (and thus
$\lambda_{OPT} \sim \Delta$), then
the algorithm finds a constant factor approximation.
\end{cor}
Note that if the graph formed by candidate edges is an expander then
the value of the following SDP solution $w_e = k/|E_{cand}|$ for each edge
$e\in E_{cand}$ is $\Omega(k/n)$, thus
$\lambda_{SDP} \geq ck/n$.
\begin{cor}
\label{cor:algconnexpander}
If the graph formed by candidate edges is an expander,
then the approximation algorithm from Theorem~\ref{thm:algconn}
finds a solution of value at least
$c \frac{k}{n\Delta}\lambda_{OPT}$.
\end{cor}
\begin{remark}
It is possible to get rid of the dependence on $\Delta$ in Theorem~\ref{thm:algconn} and
Corollary~\ref{cor:algconnexpander} and obtain approximation guarantees
of $c\min(\lambda_{OPT}, \lambda_{OPT}^2)$ and
$\frac{c k}{n} \lambda_{OPT}$ respectively. We omit the details in this extended abstract.
\end{remark}

\bibliographystyle{alpha}

\end{document}